\documentclass[final,3p,times,twocolumn]{elsarticle}

\usepackage[american]{babel}
\usepackage{epsfig}
\usepackage{color}
\usepackage{amsthm,amsfonts,amsmath,amssymb,amstext,latexsym}
\usepackage{enumitem}
\usepackage{dblfloatfix} 

\usepackage{color}
\usepackage{xspace}

\usepackage{pgfplots}
\usepackage{here}
\usepackage{wasysym}

\usepackage[scr=boondox]{mathalfa}

\usetikzlibrary{patterns}
\usetikzlibrary{automata,arrows,positioning,calc,math,arrows.meta}
\usetikzlibrary{shapes,positioning}
\usetikzlibrary{fit,intersections,decorations.pathreplacing,matrix}

\allowdisplaybreaks

\newtheorem{theorem}{Theorem}
\newtheorem{lemma}{Lemma}

\newtheorem{proposition}{Proposition}

\newtheorem{definition}{Definition}

\newcommand{\N}{\mathbb{N}}

\newcommand{\E}{\mathbb{E}}

\newcommand{\Z}{\mathbb{Z}}

\newcommand{\pp}{\mathbb{P}}

\newcommand{\sign}{\text{sign}}

\journal{Operations Research Letters}

\newcommand{\Po}{{\bf p}}
\newcommand{\Pd}{{\bf q}}

\newcommand{\ol}{p_{\ell}}
\newcommand{\dl}{q_{\ell}}
\newcommand{\od}{p_{d}}
\newcommand{\ow}{p_{w}}
\newcommand{\dd}{q_{d}}
\newcommand{\dw}{q_{w}}

\newcommand\Off{\ensuremath{\mathrm{(O)}}\xspace}
\newcommand\Def{\ensuremath{\mathrm{(D)}}\xspace}
\newcommand{\cate}{\ensuremath{\mathrm{(cat)}}\xspace}
\newcommand{\catp}{\ensuremath{\mathrm{(cat+)}}\xspace}

\begin{document}
\begin{frontmatter}
\title{
Can a Weaker Player Win? Adaptive Play in Repeated Games
}

\author{
Jonatha ANSELMI and Bruno GAUJAL\\
\{jonatha.anselmi,bruno.gaujal\}@inria.fr\\
Univ. Grenoble Alpes, Inria, CNRS, Grenoble INP, LIG, 38000 Grenoble, France.
}

\begin{abstract}
Consider a two-player game repeated $N$ times.
Player~1 can choose between two styles (for interpretability, \emph{offensive} and \emph{defensive}), whereas Player~2 uses a single fixed style.
Let $X_N:=\#\text{wins}-\#\text{losses}$ for Player~1 after $N$ games, and define the match gain as $\E[\mathrm{sign}(X_N)]$, with $\mathrm{sign}(0)=0$. We assume Player~1 is weaker in the sense that each pure style is losing in expectation.
Our objective is to identify under which parameter regimes Player~1 can nevertheless achieve a positive gain under an optimal adaptive policy.

Using dynamic programming, we solve the finite-horizon control problem and numerically identify parameter regimes in which the optimal gain is strictly positive at some horizon $N^\star$. We also derive structural conditions guaranteeing that $g_N^\star$ is always negative, and regimes (notably with fair \Def) where $g_N^\star$ is nonnegative for all $N$ and can be strictly positive for every $N\ge 2$.
We then characterize the asymptotic behavior as $N\to\infty$ for a weak player.
In the \emph{safe} case, where the defensive style induces a sure draw, the limiting gain varies continuously with the parameters and may take any value in $[0,1]$.
In the \emph{non-safe} case, the limiting gain  converges to $-1$ when both styles are strictly losing, and to $0$ when \Def\ is fair (and non-safe).

\end{abstract}



\end{frontmatter}


\section{Introduction}

\subsection{A challenge between two chess players}
\label{sec:chess}

We start with a simple yet striking situation illustrating how adaptivity can overturn a per-game disadvantage.
Two chess players face each other repeatedly, so empirical win--loss frequencies are well established. Player~1 can choose between two styles, \Off\ (offense) and \Def\ (defense), whereas Player~2 uses a single fixed style. Under either choice of Player~1, Player~2 is more likely to win. The game outcomes are:
\begin{itemize}
\item If Player~1 plays \Off, she wins with probability $0.45$ and loses with probability $0.55$ (no draws).
\item If Player~1 plays \Def, she wins with probability $0.10$, loses with probability $0.15$, and draws with probability $0.75$.
\end{itemize}
Player~1's one-game score is $+1$ for a win (W) and $-1$ for a loss (L), and her one-game expected gain is
\[
g_1 := \pp(W)-\pp(L).
\]
Thus Player~1 is weaker under both pure styles: $g_1^{\Off}=-0.10$ and $g_1^{\Def}=-0.05$. Despite this, Player~1 proposes a two-game match and we ask:
\begin{center}
\emph{Should one bet on Player~1 or on Player~2?}
\end{center}

Any non-adaptive two-game plan, i.e., $(\Off,\Off)$, $(\Off,\Def)$, $(\Def,\Off)$, or $(\Def,\Def)$, has strictly negative expected gain.
In contrast, Player~1 can achieve a positive expected gain with an adaptive rule: play \Off\ in game~1; in game~2 play \Def\ if game~1 was a win, and \Off\ otherwise.
A direct enumeration of the possible two-game outcomes yields the expected gain
\begin{align*}
g_2
&:= \pp(\text{W--W}) + \pp(\text{W--Draw}) - \pp(\text{L--L}) \\
&= 0.45\cdot 0.10 + 0.45\cdot 0.75 - 0.55\cdot 0.55 \\
&= 0.08,
\end{align*}
where outcomes with a draw in game~1 are impossible because game~1 is played under \Off.
This example illustrates a general mechanism: even if each pure style is losing in expectation, a state-dependent choice of styles over a finite horizon can yield a positive expected gain. In this paper, we study an idealized model capturing this effect.

\subsection{Model and addressed problem}

We formalize the preceding vignette as a finite-horizon repeated game between two players of unequal strength. Player~2 plays a single fixed style throughout. At each round, Player~1 chooses one of two styles: \Off\ (offense) or \Def\ (defense).
Conditional on the style chosen by Player~1 at round $n$, the round outcome is drawn independently of the past.
Under \Off, the outcome distribution is
\[
\Po = (\ow,\od,\ol), \qquad \ow+\od+\ol=1,
\]
where $\ow$ (resp.\ $\od,\ol$) is the probability that Player~1 wins (resp.\ draws, loses).
Under \Def, the distribution is
\[
\Pd = (\dw,\dd,\dl), \qquad \dw+\dd+\dl=1.
\]
We call \Def\ \emph{defensive} when draws are more likely under \Def\ than under \Off, i.e.,
\[
\dd \geq \od.
\]

In the remainder, we use the following terminology.

\begin{definition}
Player~1 is \emph{weak} if each pure style is losing in expectation, i.e.,
\[
\ow\le \ol \quad\text{and}\quad \dw\le \dl,
\]
and \emph{strictly weak} if both inequalities are strict.
\end{definition}

\begin{definition}
The defense \Def\ is \emph{safe} if $\dd=1$
and \emph{fair but non-safe} if $\dw=\dl$ and $\dd<1$.
\end{definition}

The example in Section~\ref{sec:chess} exhibits the core effect: despite strict weakness under both styles, an adaptive choice over two rounds yields a positive expected gain.

The match is played over $N$ games, and the outcome is evaluated \emph{globally} rather than game-by-game. A natural objective is the sign of the final score difference: Player~1 ``wins the match'' if the number of wins exceeds the number of losses, ``loses'' otherwise, and draws if the two are equal. Under this criterion, the relevant question is not whether Player~1 can outperform Player~2 in expectation in a single game (it cannot), but whether adaptive switching between styles over time can bias the \emph{distribution} of the final score enough to yield a positive expected match gain.

Let $W_n$ and $L_n$ denote the number of wins and losses of Player~1 after $n$ games  and let the {\it score} be $X_n := W_n-L_n$.
We evaluate the {expected } {\it gain}  at the end of the match via
\[
g_N := \E\!\left[\sign(X_N) \mid X_0 = 0\right], \qquad \sign(0)=0,
\]
so that $g_N>0$ means Player~1 has a positive expected advantage under the win/lose/draw match criterion.
The central question is whether, and for which parameter regimes, Player~1 can achieve $g_N>0$ by selecting $\Off$ or $\Def$ adaptively as a function of the match state (e.g., the current score difference), even though both pure styles are unfavorable.

This viewpoint places the interaction in the scope of finite-horizon sequential decision-making: at each round, Player~1 observes the current match state (e.g., the score difference so far) and selects a style that trades off immediate expected loss against variance and future opportunities. The key phenomenon is that, even when each pure style is unfavorable in isolation, a state-dependent mixture can exploit the finite horizon and the nonlinear match criterion to produce regimes where the weaker player has a positive chance---and even a positive expected gain---of winning the overall challenge.

\subsection{Contribution} 

In this paper, we are interested in the game above from the standpoint  of Player 1 (the weak player).
We aim at understanding under which conditions her gain is non-negative.

\begin{itemize}[leftmargin=*]

 \item
First, we clarify the role of the defensive play  parameter $\dd$. In the \emph{safe} regime $\dd=1$, we analyze the adaptive state-dependent policy \emph{catenaccio} \cate\ and obtain an explicit limit for its gain as $N\to\infty$ (Theorem~\ref{th:safeD}). In contrast, when \Def\ is {fair but non-safe} ($\dw=\dl>0$ and $\dd<1$), the same policy becomes asymptotically detrimental and its gain converges to a negative value (Theorem~\ref{th:fairD}). 

 \item
Then, we consider the \emph{optimal} adaptive policy. We formulate the problem as a finite-horizon control problem on the score state $X_n$ and solve it via Bellman optimality equation.
A numerical solution shows that $g_N^\star$ can be non-monotone in $N$ and that strict positivity may occur only at specific horizons $N^\star$ (see Figure~\ref{fig:smallN}). We complement these computations with structural results: 
(i) we provide simple sufficient conditions on $(\ow,\ol,\dw,\dl)$ under which $g_N^\star$ is \emph{uniformly negative} for all $N$, and identify a parameter regime (when \Def\ is fair) where $g_N^\star$ is nonnegative  (Proposition~\ref{prop:p2}).

 \item
Finally, we characterize the asymptotic behavior of the optimal gain as $N\to\infty$ (Theorem~\ref{th:3}).
This is the main technical difficulty and contribution.
Outside the safe defense case, the limit is \emph{quantized}: for weak players it converges either to $-1$ (both styles strictly losing) or to $0$ (defense is fair but non-safe defense). In the safe case, the limiting gain can take any value in $[0,1]$ as a function of the parameters, and the theorem also shows that a slight refinement of catenaccio, i.e., \catp, is essentially optimal in that case.
\end{itemize}

Our proofs combine random-walk arguments, finite-horizon dynamic programming, and asymptotic martingale estimates.


\section{Related Work}

The phenomenon introduced in Section~\ref{sec:chess}, i.e., combining or alternating between individually unfavorable modes of play can improve performance, is closely related to \emph{Parrondo's paradox} \cite{Harmer,ParrondoHarmerAbbott2000}, originally introduced in the context of gambling games.
In its canonical formulation, two Markovian games (or Markov chains) that are losing in isolation (negative drift)
can yield a positive long-run average gain (equivalently, positive asymptotic drift) when alternated or mixed.
This occurs because the per-step drift is state-dependent: alternating between games changes the occupation frequencies over states (e.g., capital modulo a period) and can therefore reverse the average drift.

Parrondo's paradox has found applications across disciplines, including physics, biology, and finance; see~\cite{LaiCheong2020} for a survey. The effect also extends beyond the two-game setting: variants with three or more individually losing games can still exhibit a winning outcome under suitable alternation schemes~\cite{ARENA2003545}.
Beyond random or periodic mixing,  finite-horizon optimal sequencing and adaptive selection rules have also been investigated, typically via backward induction/dynamic programming~\cite{Dinis2008,ARENA2003545}.

Our model shares this high-level mechanism, i.e., state-dependent switching between two styles.
In this sense, our setting is Parrondo-like: two individually losing modes can yield a positive outcome under switching, but here the effect is driven by finite-horizon adaptive control under a nonlinear match criterion, together with the safe/non-safe dichotomy of the defensive style.
In fact, in this paper we distinguish between the \emph{safe} (score-freezing) from \emph{non-safe} defense style of play. This yields a dichotomy that drives the long-run behavior: it enables a complete asymptotic classification for weak players and identifies regimes where \emph{catenaccio}-type policies are essentially optimal.
In addition, in contrast with the works above, our focus is not algorithmic. We are interested in structural conditions on the input parameters under which a weaker player can (or cannot) achieve a positive gain.

\section{Random walks, (un)certainty, and \emph{catenaccio}}
\label{sec:rw-catenaccio}

If Player~1 commits to a fixed style (say \Def), the match score evolves as a one-dimensional (lazy) random walk and
the probability that Player~1 wins an $N$-game match under \Def \ is (see, e.g., \cite{Feller})
\begin{equation}\label{eq:oneStyle}
\pp(X_N>0)\;=\;\sum_{\substack{i>j\ge 0\\ i+j\le N}}
\binom{N}{i}\,\binom{N-i}{j}\,
\dw^i \, \dl^j \, \dd^{N-i-j},
\end{equation}
where $i$ (resp.\ $j$) counts wins (resp.\ losses) and the remaining $N-i-j$ rounds are draws.
Therefore, if $\dw<\dl$ then $\pp(X_N>0) < 1/2$, the gain is negative for all $N$, $g^\Def_N <0$ and goes to -1 as $N$ goes to infinity.

On the other hand, with two available styles, a natural policy used by weaker teams in sports is to pursue variance early and then ``protect the lead'' once ahead. We model this as the \emph{catenaccio} policy:
\begin{center}
\emph{Play \Off\ until the score becomes positive,\\ then play \Def\ thereafter.}
\end{center}
Formally, let $\tau:=\inf\{n\ge 0:\ X_n=+1\}$ and play \Off\ up to time $\tau$ and \Def\ afterwards.

\subsection{Safe defense: a positive asymptotic gain}

When the defensive style \Def is safe, the expected gain of Player 1 under the catenaccio policy converges to  $2 \ow/\ol - 1$ as  $N$ goes to infinity.
Thus, it is positive if and only if $\ow \geq \ol/2$.
This is stated in the next result.

\begin{theorem}
\label{th:safeD}
If $\dd=1$,
then $g_N^\cate  \xrightarrow[N\to \infty]{} 2\ow/\ol - 1$.
\end{theorem}

\begin{proof}
Under \Off, the score evolves as a (possibly lazy) random walk with increments in $\{-1,0,+1\}$ and drift $\ow-\ol\le 0$. Let
$
\mathcal{W}:=\{\exists n\ge 0:\ X_n=+1\}
$
be the event that the walk ever hits $+1$. Classical ruin/hitting-probability theory yields (see \cite[p.\ 272]{Feller}),
$
\pp(\mathcal{W})= {\ow}/{\ol},
$
with $\ow\le \ol$, for the non-lazy case; the same formula holds here since laziness only inserts zero-steps.
Under catenaccio with safe \Def, hitting $+1$ at any time implies $X_N=+1$ for all subsequent horizons, whereas on $\mathcal{W}^c$ the score never becomes positive, hence $\sign(X_N)=-1$ for all $N$. Therefore,
\[
\lim_{N\to\infty} g_N^\cate
= (+1)\pp(\mathcal{W}) + (-1)\pp(\mathcal{W}^c)
= 2\frac{\ow}{\ol}-1.
\]
This concludes the proof.
\end{proof}

\subsection{Fair but non-safe defense}

Now, assume that the defensive style \Def is fair but non-safe, meaning $\dw=\dl>0$.
Intuitively, once catenaccio switches to \Def, the score continues to fluctuate and eventually forgets the early advantage.
The next result shows that the expected gain of Player 1 under the catenaccio policy always becomes negative.

\begin{theorem}
\label{th:fairD}
If $\dw=\dl>0$, $g_N^\cate   \xrightarrow[N\to \infty]{} {\ow}/{\ol}-1 <0$.
\end{theorem}

\begin{proof}
Let $\mathcal{W}$ be as in the proof of Theorem~\ref{th:safeD}; ruin theory again gives $\pp(\mathcal{W})=\ow/\ol$. On $\mathcal{W}^c$, the score never reaches $+1$, so catenaccio never switches and \Off\ is used at every round. Since $\ow<\ol$, the law of large numbers implies $X_N\to -\infty$ almost surely on $\mathcal{W}^c$, hence $\sign(X_N)\to -1$.

On $\mathcal{W}$, let $K:=\inf\{n\ge 0:\ X_n=+1\}$ be the first hitting time of $+1$. By definition of catenaccio, for all $n\ge 1$,
\[
X_{K+n}=1+\sum_{i=1}^n B_i,
\]
where $(B_i)_{i\ge 1}$ are i.i.d.\ with $\pp(B_i=+1)=\dw$, $\pp(B_i=0)=\dd$, and $\pp(B_i=-1)=\dl=\dw$. In particular, $B_i$ is mean-zero and symmetric, so $X_{K+n}-1$ is a centered (lazy) symmetric walk. By the central limit theorem (or a local limit theorem),
\[
\pp(X_{K+n}>0\mid K)\to \frac12,\qquad
\pp(X_{K+n}<0\mid K)\to \frac12,
\]
and therefore $\E[\sign(X_{K+n})\mid K]\to 0$. Since $K<\infty$ on $\mathcal{W}$, this yields
$\E[\sign(X_N)\mid \mathcal{W}]\to 0$ as $N\to\infty$.
Combining the two cases,
\begin{align*}
\lim_{N\to\infty} g_N^\cate
=&\,  \lim_{N\to\infty}\E[\sign(X_N)\mid \mathcal{W}]\,\pp(\mathcal{W}) \, +\,  \\
& \,\lim_{N\to\infty}\E[\sign(X_N)\mid \mathcal{W}^c]\,\pp(\mathcal{W}^c)\\
 =&\, 0\cdot \frac{\ow}{\ol} + (-1)\Bigl(1-\frac{\ow}{\ol}\Bigr),
\end{align*}
as desired.
\end{proof}

Therefore, when defense is merely fair but non-safe, the ``post-switch'' dynamics asymptotically erase the advantage, so catenaccio only yields transient finite-horizon benefits.

\section{Beyond \emph{catenaccio}: optimal policies and dynamic programming}
\label{sec:dp}

The \emph{catenaccio} policy provides a useful benchmark, but it is generally not optimal. In particular, when \Def\ is fair, we will see that the \emph{optimal} policy can yield a nonnegative gain for every horizon and may be strictly positive for finite $N$ even when Player~1 is weak under both pure styles. This section addresses the decision problem:
\begin{quote}
\emph{Given $\Po$ and $\Pd$, does there exist a horizon $N$ such that Player~1 can achieve $g_N(\Po,\Pd)>0$ under an optimal adaptive policy?}
\end{quote}

\subsection{Computing an optimal policy}
\label{subsec:compute}

For a fixed horizon $N$, the problem is a finite-horizon Markov decision process on the score $X_n$.
We encode the match outcome through the terminal value function
\[
V_N^N(x)=
\begin{cases}
+1, & x>0,\\
\phantom{+}0,  & x=0,\\
-1, & x<0.
\end{cases}
\]
For $0\le n\le N-1$, let $V_n^N(x)$ denote the optimal value at time $n$ when the current score is $x$. The Bellman optimality recursion reads: for all $n\ge 0$ and $x\in\{-n,\ldots,n\}$,
\begin{multline}
\label{eq:dyn}
V_n^N(x)
=\max\Bigl\{
\ow V_{n+1}^N(x+1)+\od V_{n+1}^N(x)+\ol V_{n+1}^N(x-1),\\
\dw V_{n+1}^N(x+1)+\dd V_{n+1}^N(x)+\dl V_{n+1}^N(x-1)
\Bigr\}.
\end{multline}
The optimal gain is the value at the origin:
\[
g_N^\star = V_0^N(0).
\]
In some cases, we mark the dependence on  $\Po$ and $\Pd$, by denoting the gain $g^\star_N (\Po,\Pd)$.
When the context is clear, we simply denote the gain as $g^\star_N$.
The backward computation induced by \eqref{eq:dyn} is illustrated in Figure~\ref{fig:dyn}.
\begin{figure}[h]
  \centering

\begin{tikzpicture}[
     xscale=0.8, yscale=0.8,
    emptynode/.style={circle, draw=blue!70,  fill=blue!30, inner sep=2pt},
    posnode/.style={circle, draw=green!60!black, fill=green!30, inner sep=2pt},
    negnode/.style={circle, draw=red!70!black, fill=red!40, inner sep=2pt}
]




\def\N{6}
\def\xoff{3}
\def\yoff{4}

\foreach \x in {0,...,\N} {
     \foreach \y in {0,...,\x} {

               \ifnum\x<\N
            \draw ({\x},{\y})
                  -- ({\x+1},{\y+1});
            \draw  ({\x},{\y})
                  -- ({\x+1},{\y-1});
                  \fi

               \ifnum\x<\N
            \draw ({\x},{\y})
                  -- ({\x+1},{\y});
                  \fi
                     \node[emptynode] at ({\x},{\y}) {};
  }
    \foreach \y in {0,...,\x} {

               \ifnum\x<\N
            \draw ({\x},{-\y}) -- ({\x+1},{-\y+1});
            \draw ({\x},{-\y}) -- ({\x+1},{-\y-1});
            \fi

               \ifnum\x<\N
            \draw ({\x},{-\y})
                  -- ({\x+1},{-\y});
        \fi
                 \node[emptynode] at ({\x},{-\y}) {};
      }

    \foreach \x in {4,...,\N} {
     \foreach \y in {4,...,\x} {
       \node[posnode] at ({\x},{\y}) {};
     }

     \foreach \y in {4,...,\x} {
              \node[posnode] at ({\x},{-\y+7}) {};
      }
    }

    \foreach \x in {4,...,\N} {
     \foreach \y in {4,...,\x} {
       \node[negnode] at ({\x},{-\y}) {};
     }

     \foreach \y in {4,...,\x} {
              \node[negnode] at ({\x},{\y-7}) {};
      }
    }

\foreach \y in {1,...,6} {
    \node[posnode] at ({\N},{\y}) {};
    \node[right, green!60!black] at ({\N+0.4},{\y}) {$+1$};
  }

\foreach \y in {1,...,6} {
    \node[negnode] at ({\N},{-\y}) {};
    \node[right, red!60!black] at ({\N+0.4},{-\y}) {$-1$};
  }

\node[emptynode] at ({\N},{0}) {};
\node[right,blue] at ({\N+0.6},{0}) {$0$};

  }


\foreach \x in {0,...,\N} {
  \node[below, font=\small] at (\x,-0.12) {\x};   
}

\end{tikzpicture}


\caption{Backward computation of $g_N^\star=V_0^N(0)$ for $N=6$. Terminal values are $\sign(X_6)$.
Green (resp.\ red) nodes are states from which a match win (resp.\ loss) is already forced, hence their values are $+1$ (resp.\ $-1$).}

  \label{fig:dyn}
\end{figure}
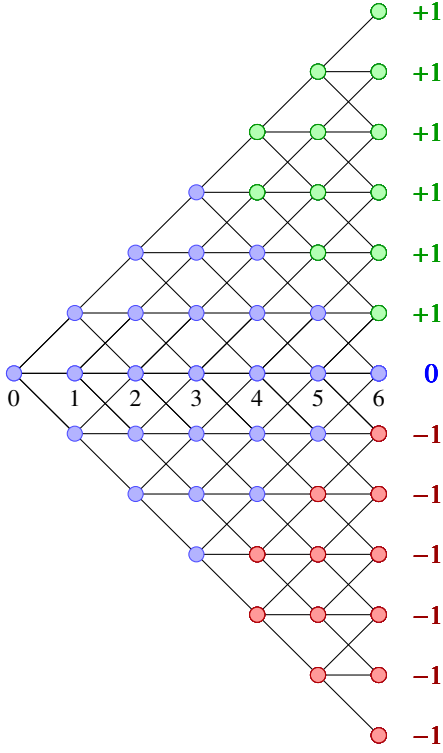
The lattice points $(n,x)$ shown in red or green need not be evaluated. Green nodes have value $+1$ (from such states Player~1 cannot lose, regardless of future actions), whereas red nodes have value $-1$ (Player~1 cannot win). Exploiting these forced regions eliminates roughly half of the states from the backward recursion.

\subsection{Properties of the optimal gain}
\label{subsec:props}

Numerical solutions of \eqref{eq:dyn} reveal two non-intuitive effects: (i) $g_N^\star$ can be \emph{non-monotone} in $N$; (ii) the optimal action is highly horizon-dependent. Figure~\ref{fig:smallN} illustrates this phenomenon for strictly losing parameters
$\ow = 0.43,\ \od = 0,\ \ol = 0.57$ and $\dw = 0.06,\ \dd = 0.86,\ \dl = 0.08,$
where the optimal gain is positive at $N^\star=4$ but not for all nearby horizons.

\begin{figure}[hbtp]
\includegraphics[width=1.0\linewidth]{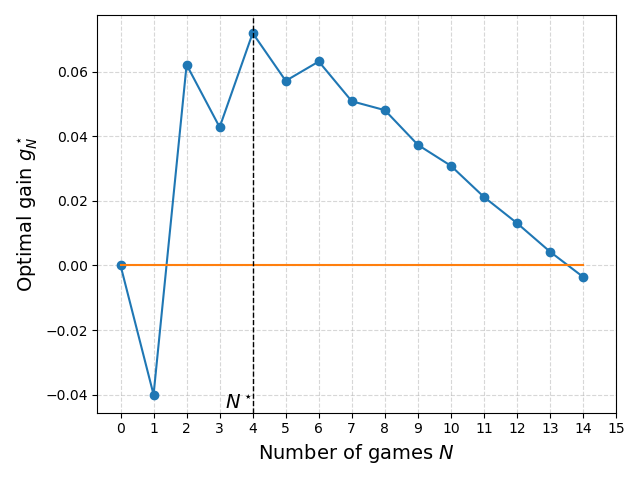}
\caption{Optimal gain $g_N^\star$ under the strictly losing parameters $(\ow ,\od, \ol) = (0.43,0,0.57)$ and $(\dw,\dd,\dl) = ( 0.06,0.84,0.10)$.}
\label{fig:smallN}
\end{figure}

Identifying an optimal horizon $N^\star$ as a function of $(\ow,\od,\ol,\dw,\dd,\dl)$ appears delicate because $N^\star$ can be large when both styles are only mildly losing.
For example, with $(\ow ,\od, \ol) = (0.49,0,0.51)$ and $(\dw,\dd,\dl) = ( 0.02,0.95,0.03)$, we obtain $N^\star=32$ and $g_{32}^\star\approx 0.453$.
It is not difficult to show that $N^\star$ can grow to infinity by making \Def safer and safer (see Section~\ref{subsec:asympt}).

Therefore, deciding whether a weak player can achieve $g_N^\star>0$ for some horizon $N$ is non-trivial.
When draws are rare, odd horizons often yield a smaller gain than the preceding even horizon.
The next proposition formalizes this no-draw disadvantage. Note that this is not true in general: when $\dd>0$ and $\od=0$, it is possible to construct an example where $g^\star_{2N+1}>g^\star_{2N}$ (e.g., by taking $\Po=(0.4,0,0.6)$ and $\Pd=(0.15,0.7,0.15)$).

\begin{proposition}
\label{prop:p1}
If $\od=\dd=0$, then
$g^\star_{2N+1}\le g^\star_{2N}$.
\end{proposition}

\begin{proof}
  Let us consider a $(2N+1)$-game match.
With $\od=\dd=0$, the score parity matches the time index: after $2N$ games, $X_{2N}$ is even. 
If $X_{2N}\ge 2$ (resp.\ $X_{2N}\le -2$), the final outcome is forced regardless of the last action, hence
$V^{2N+1}_{2N}(X_{2N})=+1$ (resp.\ $-1$).
If $X_{2N}=0$, the last game yields expected value $\max(\ow-\ol,\dw-\dl)<0$ under weakness, so
$V^{2N+1}_{2N}(0)<0$.
In all cases,
\[
V^{2N+1}_{2N}(x)\le V^{2N}_{2N}(x)\qquad\forall x\in\{-2N,\ldots,2N\},
\]
since $V^{2N}_{2N}(x)=\sign(x)$.

Let $\pi^\star$ be optimal for horizon $2N\!+\!1$. Then
\begin{align*}
g^\star_{2N+1}
&= \E_{\pi^\star}\!\left[V^{2N+1}_{2N+1}(X_{2N+1})\,\middle|\,X_0=0\right] \\
& = \E_{\pi^\star}\!\left[V^{2N+1}_{2N}(X_{2N})\,\middle|\,X_0=0\right] \\
&\le \E_{\pi^\star}\!\left[V^{2N}_{2N}(X_{2N})\,\middle|\,X_0=0\right]\le g^\star_{2N}.
\end{align*}
This concludes the proof.
\end{proof}

Let us recall that with no loss of generality we assume $\od \leq \dd$ (draws are more likely when playing \Def).
Note that this implies $\dw \leq \ol$ in the weak player case.

Let us consider two defensive strategies with different probabilities, denoted \Def  and \Def' with $\Pd$ and $\Pd' = (\dw',\dd',\dl')$, respectively.
We say that \Def' {\it dominates} \Def if the corresponding probabilities are stochastically ordered: $\Pd \leq_{st} \Pd'$; this simply means $\dw \leq \dw'$ and $\dl \geq \dl'$.
It should be clear that the gain under \Def' is better than the gain under \Def.
\begin{lemma}\label{lem:dominates}
  If $\Pd \leq_{st} \Pd'$ then $g^\star_N(\Po,\Pd) \leq g^\star_N(\Po,\Pd')$.
\end{lemma}
\begin{proof}
  The proof holds by backward induction on $n$.
  We show that for all $x$, the respective values satisfy $V^N_n(x) \leq V^{N'}_n(x)$, $V^N_n(x)$ and $V^N_n(x)$ are non-decreasing in $x$.
  When $n=N$, $V^N_N(x) = V^{N'}_n(x)$ and they are both non-decreasing by definition.
  The Bellman equation at $n$ for the original probabilities $\Po,\Pd$ is:
\begin{small}
\begin{align*}
  V^N_n(x)  = & \max \Big\{  \ow V^N_{n+1}(x+1)  + \od V^N_{n+1}(x)  + \ol V^N_{n+1}(x-1), \\
  &\dw V^N_{n+1}(x+1) + \dd V^N_{n+1}(x) + \dl V^N_{n+1}(x-1) \Big\}\\
  \leq & \max \Big\{  \ow V^N_{n+1}(x+1)  + \od V^N_{n+1}(x)  + \ol V^N_{n+1}(x-1), \\
  & \dw' V^N_{n+1}(x+1) + \dd' V^N_{n+1}(x) + \dl' V^N_{n+1}(x-1) \Big\}\\
      \leq & \max \Big\{  \ow V^{N'}_{n+1}(x+1)  + \od V^{N'}_{n+1}(x)  + \ol V^{N'}_{n+1}(x-1), \\
  & \dw' V^{N'}_{n+1}(x+1) + \dd' V^{N'}_{n+1}(x) + \dl' V^{N'}_{n+1}(x-1) \Big\}
  =  V^{N'}_n(x).
\end{align*}
\end{small}
The first inequality comes from  $V^N_{n+1}(x)$ being  non-decreasing in $x$ by induction, so that
 $V^N_{n+1}(x-1) \leq V^{N}_{n+1}(x)\leq V^{N}_{n+1}(x+1)$. Using $\Pd \leq_{st} \Pd'$ implies the first inequality.
The second inequality comes from  induction $V^N_{n+1}(x-1) \leq V^{N'}_{n+1}(x-1)$, $V^N_{n+1}(x) \leq V^{N'}_{n+1}(x)$, $V^N_{n+1}(x+1) \leq V^{N'}_{n+1}(x+1)$.
The fact that $V^N_n(x)$ and $V^{N'}_n(x)$ are non-decreasing in $x$ comes from the monotoniticy of the Bellman operator and because $V^N_n(x)$ and $V^{N'}_n(x)$ are non-decreasing in $x$. This concludes the induction.
Finally, $ g^\star_N(\Po,\Pd)  = V^N_0(0) \leq V^{N'}_0(0 )= g^\star_N(\Po,\Pd')$.
\end{proof}


In the following, we use the notation $x\lor y := \max(x,y)$ and  $x\land y := \min(x,y)$.

\begin{proposition} \label{prop:p2} Consider a weak player. Then 
  \begin{itemize}
  \item[(i)] If $\dl \geq \ow$, then for all $N\geq 1$, $g^\star_N \leq 0$  
  \item[(ii)] If $\dw = \dl$ (\Def is fair)  then the gain is non-negative:   $g^\star_N \geq 0$. If in addition $\dl=\dw < \ow$ (\Def does not dominate \Off), then   $g^\star_N >0$, for all $N \geq 2$.
  
  \end{itemize}
\end{proposition}
\begin{proof}
  
  Proof of (i).
Let us define  \Def' with probabilities $\dw' := \dw\lor \ow, \dl' := \dl \land \ol, \dd' := 1-\dw'-\dl'$. By definition, $\Pd' $ dominates $\Pd$.
Therefore, by using Lemma \ref{lem:dominates}, $g^\star_N(\Po,\Pd) \leq  g^\star_N(\Po,\Pd')$.
By construction,  $\Pd' $ also dominates $\Po$ so  $g^\star_N(\Po,\Pd') \leq  g^\star_N(\Pd',\Pd')$.

 Now in the case where the offensive and defensive choices are the same, the gain is easier  to analyse.
 The gain of \Def' in one game  is equal to $g^{\Def'}_1 = \dw'-\dl' = \dw\lor \ow -  \dl \land \ol \leq 0$ because $\dl \geq  \ow$ and $\dw \leq \dl$
 For any $N\geq 1$, the probability $\pp(X_N > 0)$ is given by Eq. \eqref{eq:oneStyle} (replacing $\Pd$ by $\Pd'$). This is always smaller than $1/2$, and  goes to 0 as $N$ goes to infinity. So the gain is negative and goes to $-1$.



  Proof of (ii).
  Since \Def is fair, playing \Def all the time will give a zero gain, hence $g^\star_N \geq 0$.
  In addition, a case analysis shows that when $\dw= \dl < \ow$, the optimal policy uses \Off after game $N-1$ in state -1 so that  $g^\star_N>0$ since
  $g^\star_N =  \E_{\pi^\star} (V^{N}_{N-1}(X_{N-1}) | X_0 = 0)$ is strictly increasing in $V^{N}_{N-1}(-1)$ and since score $-1$ can be reached at step $N-1$ with positive probability.
  This last point also shows that in general, if no playing style dominates the other, then the gain of the optimal policy is strictly greater than the gain obtained by playing a fixed style.
\end{proof}

The previous proposition has interesting consequences. In particular, if a draw is not possible ($\od = \ol = 0$), then the gain is always negative for a strictly weak player.
It can also be understood on the positive side:
When the player is strictly weak,  the only configuration where her gain has a chance to be positive is when
$\dw < \dl \leq \ow < \ol$, although  this is not sufficient to guarantee a positive gain.

\subsection{Asymptotics}
\label{subsec:asympt}

Unlike the finite-horizon case, the asymptotic behavior as $N\to\infty$ admits a complete classification under the weak-player assumption. If both \Off\ and \Def\ are strictly losing, the match is eventually lost with overwhelming probability and $g_N^\star\to -1$. If \Def\ is fair but non-safe, then $g_N^\star\ge 0$ for every $N$ (by always playing \Def), yet $g_N^\star\to 0$. Finally, if \Def\ is safe ($\dd=1$), then $g_N^\star$ is nondecreasing in $N$ and converges to the value induced by (essentially) catenaccio.

These results are formally given in the next theorem, which uses \emph{catenaccio+} (denoted \catp), a slight refinement of \cate.
Catenaccio+ is defined as follows.
It coincides with catenaccio, except in one case: if after $N-1$ games the score is zero, $X_{N-1}=0$ (so the last game determines whether Player~1 wins or loses), then catenaccio+ selects the better of \Off\ and \Def. Specifically, it chooses \Off\ if $\ow-\ol>\dw-\dl$, and \Def\ otherwise. This differs from catenaccio, which plays \Def\ if the score has ever been positive, and \Off\ if the score has remained non-positive. Clearly, catenaccio+ is (weakly) better than catenaccio.

\begin{theorem}
\label{th:3}
Consider a weak player. 
\begin{itemize}
\item[(i)] If $\ow < \ol$ and $\dw < \dl$, then $g_N^\star \xrightarrow[N\to\infty]{} -1$.
\item[(ii)] If $\ow < \ol$ and $\dw = \dl > 0$, then $\forall N \geq 1, g_N^\star \ge 0$ and $g_N^\star \xrightarrow[N\to\infty]{} 0$.
\item[(iii)] If $\ow < \ol$ and $\dd = 1$, then $g_N^\star= 0 \lor \max_{n \le N} g_n^{\catp} $ and  
$g_N^\star \xrightarrow[N\to\infty]{} 0 \lor \Bigl(2\frac{\ow}{\ol} - 1\Bigr)$.
\end{itemize}
\end{theorem}

\begin{proof}
\noindent\emph{(i).}
Let $\pi$ be an arbitrary policy (defined for all $n\in\N$ and $x\in\Z$, with $\pi(x,n)\in\{\Off,\Def\}$), and let $(X_n)_{n\ge 0}$ be the resulting score process with natural filtration $(\mathcal F_n)_{n\ge 0}$.
Set 
$\delta := \ol-\ow \land \dl-\dw > 0$.
By definition of the game, whatever action is taken at round~$n$,
\[
\E[X_{n+1}\mid \mathcal F_n] \le X_n-\delta .
\]
Define $Y_n:=X_n+\delta n$. Then $(Y_n)$ is a supermartingale:
\[
\E[Y_{n+1}\mid \mathcal F_n]\le Y_n
\]
with bounded increments as
\[
|Y_{n+1}-Y_n|
=|(X_{n+1}-X_n)+\delta|
\le 1+\delta ,
\qquad Y_0=0.
\]
Therefore, Azuma--Hoeffding for supermartingales yields \cite[Corollary~2.20]{Wainwright2019HDS}
\[
\pp(Y_n>h)\le \exp\!\left(-\frac{h^2}{2n(1+\delta)^2}\right),
\]
for any $h>0$.
Taking $h=(\delta n)^{2/3}$ and returning to $X_n$ gives
\begin{equation}\label{eq:az}
\pp\!\left(X_n>(\delta n)^{2/3}-\delta n\right)
\le \exp\!\left(-\frac{(\delta n)^{1/3}}{2(1+\delta)^2}\right).
\end{equation}
Since $(\delta n)^{2/3}-\delta n\to -\infty$, it follows that $\pp(X_n<0)\to 1$, hence $g_n^\pi\to -1$ as $n\to\infty$.

\medskip
\noindent\emph{(ii).}
Clearly $g_N^\star\ge 0$ for all $N$, since the stationary policy that plays \Def\ at every step yields gain $0$.

To prove $g_N^\star\to 0$, fix an arbitrary policy $\pi$ and let $X_n$ be the induced score process. Here, 
\[
\E[X_{n+1}\mid \mathcal F_n]\le X_n.
\]
Thus $(X_n)$ is a supermartingale with increments bounded by $1$.  However, Azuma--Hoeffding
is not sufficient to conclude that $\pp(X_n>0)$ falls below $1/2$.

Let us refine the construction and let $K_n$ be the number of times $\pi$ uses \Off\ in the first $n$ steps.
Here, the idea is to decompose
\[
X_n=\sum_{i=1}^{K_n} O_i+\sum_{j=1}^{n-K_n} D_j,
\]
where $(O_i)_{i\in\N}$ are i.i.d.\ with values $+1$ w.p.\ $\ow$, $0$ w.p.\ $\od$, and $-1$ w.p.\ $\ol$, and $(D_j)_{j\in\N}$ are i.i.d.\ with values $+1$ w.p.\ $\dw$, $0$ w.p.\ $\dd$, and $-1$ w.p.\ $\dl=\dw$. Since $\pi$ depends on $(n,X_n)$, $K_n$ is a random variable depending on the realized samples $\{O_i,D_j:i,j\le n\}$.

Consider the event $ \mathcal S:=\{(O_i)_{i \in \N}, (D_j)_{j\in \N} s.t.  K_n=o(\sqrt n)\}$.
First, let us assume first that $\mathcal S$ holds. Since $-1\le O_i\le 1$,
\begin{equation}\label{eq:bound}
\sum_{j=1}^{n-K_n}D_j-K_n \le X_n \le \sum_{j=1}^{n-K_n}D_j+K_n .
\end{equation}
Using also $-1\le D_j\le 1$ yields
\begin{equation}\label{eq:bound2}
\sum_{j=1}^{n}D_j-2K_n \le X_n \le \sum_{j=1}^{n}D_j+2K_n .
\end{equation}
By the central limit theorem,
\[
\frac{1}{\sqrt n}\sum_{j=1}^{n} D_j \Rightarrow \mathcal N(0,\dw).
\]
Dividing \eqref{eq:bound2} by $\sqrt n$ and using $K_n/\sqrt n\to 0$ on $\mathcal S$, we obtain
\[
\pp(X_n>0)\to \tfrac12,
\qquad
\pp(X_n<0)\to \tfrac12,
\]
hence $g_n^\pi\to 0$ on $\mathcal S$.

Assume now that $\mathcal S^c$ holds. Then $K_n$ is of order $\sqrt n$ or larger, and $\sum_{i=1}^{K_n} O_i$ has negative drift. Using the same Azuma--Hoeffding argument as in \eqref{eq:az} together with a union bound yields $\pp(\sum_{i=1}^{K_n}O_i>0)\to 0$. Indeed, for $n$ large enough $K_n\ge \sqrt n$, and with $h=K_n^{2/3}$,
\begin{align*}
&\pp\!\left(\sum_{i=1}^{K_n} O_i >
K_n^{2/3}-(\ol-\ow)K_n\right)\\
&\le \sum_{k=\sqrt n}^{n}
\pp\!\left(\sum_{i=1}^{k} O_i >
k^{2/3}-(\ol-\ow)k\right)  \\
&\le \sum_{k=\sqrt n}^{n}\exp\!\left(-\frac{k^{4/3}}{2k}\right)
\le n\exp\!\left(-\frac{n^{1/6}}{2}\right).
\end{align*}
Consequently,
\begin{align*}
\pp\!\left(\sum_{i=1}^{K_n} O_i>0\right)
& \le
\pp\!\left(\sum_{i=1}^{K_n} O_i >
n^{1/3}-(\ol-\ow)\sqrt n\right) \\
& \le
\pp\!\left(\sum_{i=1}^{K_n} O_i >
K_n^{2/3}-(\ol-\ow)K_n\right)  \\
& \le n\exp\!\left(-\frac{n^{1/6}}{2}\right),
\end{align*}
which vanishes as $n\to\infty$. By symmetry, $\pp(\sum_{j=1}^{n-K_n}D_j>0)\le 1/2$, hence for large $n$ we have $\pp(X_n>0)<1/2$ and therefore $g_n^\pi<0$ on $\mathcal S^c$.

Combining the cases $\mathcal S$ and $\mathcal S^c$ gives
\[
\limsup_{n\to\infty} g_n^\pi \le 0
\]
for every policy $\pi$. Since $g_N^\star\ge 0$ for all $N$, it follows that $g_N^\star\to 0$ as $N\to\infty$.

\medskip
\noindent\emph{(iii).}
First, $g_{N+1}^\star\ge g_N^\star$ when $N \geq 1$ because one can use the following policy   for $N+1$ steps:
{\it Use \Def at the first step and that use the optimal policy for the remaining $N$ steps.}
Since \Def is safe, the first game does not change the score. The gain of this policy is the gain of the optimal policy for $N$ games.

Now, let us prove that
\begin{equation}\label{eq:opt=cat}
g_N^\star = 0\lor \max_{n\le N} g_n^{\catp}.
\end{equation}
The inequality $g_N^\star \ge 0\lor \max_{n\le N} g_n^{\catp}$ is immediate: the optimal policy dominates the stationary \Def\ policy (gain $0$) and, for each $n\le N$, the policy that plays \Def\ for the first $N-n$ rounds and then plays catenaccio+ for $n$ rounds.

We now show $g_N^\star \le 0\lor \max_{n\le N} g_n^{\catp}$ by induction on $N$.
For $N=1$, catenaccio+ is optimal.
Fix $N\ge 2$ and consider the score after the first round under an optimal policy.
If the optimal policy plays \Def\ at step~1, then $X_1=0$ and the gain reduces to $g_{N-1}^\star$, which is bounded by $0\lor \max_{n\le N} g_n^{\catp}$ by the induction hypothesis.
If the optimal policy plays \Off\ at step~1, we distinguish the cases $X_1\in\{-1,0,1\}$:
\begin{itemize}
 \item
If $X_1=1$, then playing \Def\ thereafter guarantees a match win; this coincides with catenaccio+, hence the conditional gain is bounded by $0\lor \max_{n\le N} g_n^{\catp}$: $(g^\star_N|X_1=1) = +1 = (g^\catp_N|X_1=1)  \leq  0 \lor  \max_{n\leq N} (g^\catp_n|X_1=1)$

 \item
If $X_1=0$, we reduce to horizon $N-1$ and apply the induction hypothesis:
$(g^\star_N|X_1 = 0) =   g^\star_{N-1} \leq 0\lor   \max_{n\leq N-1} g^\catp_n \leq 0\lor \max_{n\leq N} g^\catp_n$.

 \item
If $X_1 = -1$ then, there are two sub-cases: if score $0$ is not reached, then the gain is $-1  \leq \max_{n\leq N} g^\catp_n$.
If score $0$ is reached at some step $K$,  then we are back to the problem with $N-K$ games and induction applies. The gain $(g^\star_N| X_1 = -1) = g^\star_{N-K}  \leq 0\lor \max_{n\leq N-K} g^\catp_n \leq 0 \lor \max_{n\leq N} g^\catp_n$.

\end{itemize}


Finally, $\catp$ differs from \cate\ only on the event $\{X_{N-1}=0\}$, whose probability under either policy vanishes as $N\to\infty$; hence
$\lim_{N\to\infty} g_N^{\catp}=\lim_{N\to\infty} g_N^{\cate}$.
By Theorem~\ref{th:safeD}, $\lim_{N\to\infty} g_N^{\cate}=2\ow/\ol-1$, and the coincidence of the (nonnegative) limits of $g_N^\star$ and $g_N^{\catp}$ follows directly from \eqref{eq:opt=cat}.
\end{proof}

Theorem~\ref{th:3} yields a complete asymptotic classification for a weak player: the asymptotic gain is always in $\{-1,0\}$ outside the safe regime, whereas safe defense yields a continuum of limits in $[0,1]$ through the parameter $\ow/\ol$. Point~(iii) also formalizes that \cate/\catp\ is essentially optimal when \Def\ is safe, but becomes asymptotically suboptimal as soon as \Def\ is merely fair and non-safe.



\section{Conclusion}

This paper identifies regimes in which an apparently weaker player can nevertheless attain a positive expected match outcome by adapting between two styles of play.
The optimal advantage can be non-monotone in the number of games, and positivity may occur only at specific horizons.
Finally, we provide a complete long-run classification for weak players.

Extending these insights to richer match structures, e.g., multiple outcome levels and multi-point scoring rules, is a natural direction for future work.

\section{Second conclusion}

To conclude on a sociological note, several behavioral studies suggest that the defeat of an apparently stronger player may generate frustration that can escalate into aggressive behavior.
In such situations, the wiser strategy for the other player may be the one suggested by C-3PO to R2-D2 during a holographic chess game against Chewbacca:
``Let the Wookiee win.''

\section*{Acknowledgments}

This work has been partially supported by the MIAI cluster ANR-23-IACL-0006.

\bibliographystyle{abbrv}
\bibliography{references.bib} 

\end{document}